\documentclass[conference,letterpaper]{IEEEtran}
\usepackage[utf8]{inputenc} 
\usepackage[T1]{fontenc}
\usepackage{url}
\usepackage{ifthen}
\usepackage{cite}
\usepackage[cmex10]{amsmath} 
\usepackage{amsfonts}
\usepackage{amssymb}
\usepackage{amsthm}
\usepackage{hyperref}
\usepackage[dvipsnames]{xcolor}
\usepackage{mathtools}
\usepackage{bm}
\usepackage{lipsum}  
\usepackage{flushend}

\usepackage{tikz}
\usetikzlibrary{shapes.misc,positioning,bayesnet,matrix}

\newtheorem{lemma}{Lemma}
\newtheorem{proposition}{Proposition}
\newtheorem{corollary}{Corollary}

\newcommand{\E}[2][]{\mathbb{E}_{#1}\left[#2\right]}
\DeclarePairedDelimiterX{\infdivx}[2]{(}{)}{
  #1\;\delimsize\Vert\;#2
}
\newcommand{\KL}{D\infdivx}

\newcommand\myshade{85}
\colorlet{mylinkcolor}{BrickRed}
\colorlet{mycitecolor}{NavyBlue}
\colorlet{myurlcolor}{Aquamarine}

\hypersetup{
  linkcolor  = mylinkcolor!\myshade!black,
  citecolor  = mycitecolor!\myshade!black,
  urlcolor   = myurlcolor!\myshade!black,
  colorlinks = true,
}

\begin{document}
\title{Characterising directed and undirected metrics\\ of high-order interdependence}

\author{
  \IEEEauthorblockN{Fernando E. Rosas}
  \IEEEauthorblockA{Department of Informatics\\
                    University of Sussex\\
                    Brighton BN1 9QJ, UK\\
                    \small\texttt{f.rosas@sussex.ac.uk}}
  \and
  \IEEEauthorblockN{Pedro A.M. Mediano}
  \IEEEauthorblockA{Department of Computing\\
                    Imperial College London\\
                    London SW7 2RH, UK\\
                \small\texttt{p.mediano@imperial.ac.uk}}
  \and
  \IEEEauthorblockN{Michael Gastpar}
  \IEEEauthorblockA{School of Computer and Communication Sciences\\
                    École Polytechnique Fédérale de Lausanne\\
                    1015 Lausanne, Switzerland\\
                    \small\texttt{michael.gastpar@epfl.ch}}
}

\maketitle

\begin{abstract}
Systems of interest for theoretical or experimental work often exhibit high-order interactions, corresponding to statistical interdependencies in groups of variables that cannot be reduced to dependencies in subsets of them. While still under active development, the framework of partial information decomposition (PID) has emerged as the dominant approach to conceptualise and calculate high-order interdependencies. 
PID approaches can be grouped in two types: 
\textit{directed} approaches that divide variables into sources and targets, and \textit{undirected} approaches that treat all variables equally. Directed and undirected approaches are usually employed to investigate different scenarios, and hence little is known about how these two types of approaches may relate to each other, or if their corresponding quantities are linked in some way. In this paper we investigate the relationship between the redundancy-synergy index (RSI) and the O-information, which are practical metrics of directed and undirected high-order interdependencies, respectively. 
Our results reveal tight links between these two quantities, and provide two interpretations of them in terms of likelihood ratios in a hypothesis testing setting, as well as in terms of projections in information geometry.

\end{abstract}

\section{Introduction}

Variables characterising biological, social, or technological systems can exhibit complex statistical relationships. 
Crucially, while bivariate interdependencies are fully characterised by their strength (measured e.g.  by the mutual information), higher-order relationships involving three or more variables can be of qualitatively different kinds --- most notably redundant (where multiple variables share the same information) or synergistic (where a set of variables holds some information that cannot be seen from any subset).
The assessment of high-order interdependencies has lead to powerful practical analyses, revealing a plethora of new insights into, for example, the inner workings of natural evolution~\cite{rajpal2023quantifying}, genetic information flow~\cite{cang2020inferring}, psychometric interactions~\cite{marinazzo2022information,varley2022untangling}, cellular automata~\cite{rosas2018information},
and baroque music~\cite{scagliarini2022quantifying}. 
The analysis of higher-order interdependencies has been particularly fruitful for studying biological~\cite{gatica2021high,luppi2022synergistic,herzog2022genuine,varley2023information,varley2023multivariate} and artificial~\cite{tax2017partial,proca2022synergistic,kaplanis2023learning} neural systems, where redundancies and synergies have been found to play different roles balancing the needs for reliable sensory-motor actions on the one hand, and flexible and adaptive activity in association cortices on the other~\cite{luppi2024information}.

Partial information decomposition (PID) has emerged as the dominant framework to reason about high-order statistical interdependencies~\cite{williams2010nonnegative,griffith2014quantifying,wibral2017partial,mediano2021towards}. 
PID approaches can be divided into two families: \emph{directed} approaches that decompose information provided by source variables about a target (as described in the original formulation of PID~\cite{williams2010nonnegative}), and \emph{undirected} approaches that disentangle the information shared between variables without attributing distinct roles to them~\cite{rosas2016understanding,ince2017partial}, being invariant under permutations. 
These approaches find natural application in different scenarios --- directed approaches are natural for investigating e.g. neural systems with clear inputs and outputs~\cite{wibral2017partial}, while undirected approaches are well suited to study interdependencies between equivalent entities, e.g. spins in Ising models~\cite{vijayaraghavan2017anatomy}. We note that a different type of directed information analysis can be found, e.g. in Ref.~\cite{So_2012}.

Despite its mathematical elegance, an important weakness of PID is that the number of its atoms scales super-exponentially, and hence the computation of the full decomposition is unfeasible for analyses involving a large number of variables. 
This limitation can be overcome by using so-called \textit{coarse-grained PID metrics}, which calculate not individual atoms but linear combinations of them. 
Such metrics also adhere to the directed/undirected distinction: for example, the overall balance between redundancy and synergy can be estimated in a directed manner via the redundancy-synergy index (RSI)~\cite{nips02-NS02}, and in an undirected manner via the O-information~\cite{rosas2019quantifying}. 

Contrasting directed and undirected PID approaches is a fertile, yet understudied, way to deepen our understanding of information decomposition.
As a step in this direction, in this paper we present an analysis of the relationship between the RSI and O-information in various ways. 
Our analysis reveals analytical expressions for the RSI and O-information in terms of each other, which highlights similarities and differences between these metrics. Additionally, our results lead to a novel characterisation of the RSI and O-information in terms of the log-likelihood ratio between redundancy and synergy-dominated models, providing an interpretation of these quantities in a hypothesis testing setting.

\section{Definitions}
\label{sec:preliminaries}

We consider a multivariate system of $n$ joint random variables, with a state specified by the random vector $\bm X=(X_1,\dots,X_n)$. We will investigate various metrics that assess the interdependencies between the parts of $\bm X$, sometimes considering how $\bm X$ relates to another variable $Y$. 
We use the shorthand notation $\bm X_i^j=(X_i,\dots,X_j)$ for $1\leq i\leq j\leq n$, $\bm X^j = \bm X_1^j$, and $\bm X_{-j} = (X_1,\dots,X_{j-1},X_{j+1},\dots,X_n)$.
For simplicity we focus on the case of discrete variables, although most of our results apply directly to the continuous case. 
The alphabet of $X_i$ and $Y$ is denoted by $\mathcal{X}_i$ and $\mathcal{Y}$, respectively, and their cardinality by $|\mathcal{X}_i|$ and $|\mathcal{Y}|$.

Let us now define the two metrics that are at the focus of our investigation. The \textit{O-information}~\cite{rosas2019quantifying} of the components of the random vector $\bm X$ is given by
\begin{equation}
    \Omega(\bm X) \coloneqq (n-2) H(\bm X) + \sum_{j=1}^n \Big( H(X_j) - H(\bm X_{-j}) \Big),\label{eq:def_Oinfo}
\end{equation}
where $H$ is the Shannon entropy. 
The \textit{redundancy-synergy index}~\cite{nips02-NS02} of the components of the random vector $\bm X$ and the random variable $Y$ is given by
\begin{equation}
    \text{RSI}(\bm X; Y) \coloneq \sum_{j=1}^n I(X_j; Y) - I(\bm X;Y),\label{eq:def_RSI}
\end{equation}
where $I$ denotes Shannon's mutual information. 
Note that both the O-information and the RSI are symmetric to permutations in the ordering of the components of $\bm X$; however, the $\text{RSI}$ is not invariant to exchanging the role of $Y$ and any of the $X_i$'s. 

A fundamental property of the O-information and RSI is that both quantify the balance between synergistic and redundant interdependencies. As an example, consider the `three-bit copy' system, where $X_1$ is a fair coin flip and $X_3 = X_2 = X_1$); and the \texttt{xor} logic gate, where $X_1,X_2$ are independent fair coin flips and $X_3 = X_1 + X_2 \text{(mod 2)}$. Both RSI and O-information are positive for the copy system, representing redundancy; and negative for the \texttt{xor}, representing synergy. Although both metrics yield identical results for these small systems, they in general differ. More details about how these metrics reflect high-order interdependencies can be found in Refs.~\cite{rosas2019quantifying,timme2014synergy}.

Our analyses will use three different multivariate extensions of Shannon's mutual information, specifically the \textit{Total Correlation} (TC)~\cite{watanabe1960information}, \textit{Dual Total Correlation} (DTC)~\cite{te1978nonnegative}, and \textit{Interaction Information} (II)~\cite{mcgill1954multivariate},\footnote{The interaction information is
closely related to the \textit{I-measures}~\cite{yeung1991}, the
\textit{co-information}~\cite{Bell2003}, and the \textit{multi-scale
complexity}~\cite{bar2004multiscale}.} defined as follows:
\begin{align}
    \text{TC}(\bm X) \coloneqq& \,\sum_{j=1}^n H(X_j) - H(\bm X), \label{eq:def_TC}\\
    \text{DTC}(\bm X) \coloneqq& \: \,H(\bm X) -\sum_{j=1}^n H(X_j|\bm X_{-j}),\label{eq:def_DTC}
    \\
    \text{II}(\bm X) \coloneqq&  \sum_{\bm{\gamma}\subseteq\{1,\dots,n\}} (-1)^{|\bm{\gamma}|+1} H(\bm{X}^{\bm{\gamma}}),\label{eq:def_II}
\end{align}
where the sum in Eq.~\eqref{eq:def_II} goes over all the subsets of indices $\bm{\gamma}\subseteq\{1,\dots,n\}$,
with $|\bm{\gamma}|$ being the cardinality of $\bm{\gamma}$ and
$\bm{X}^{\bm{\gamma}}$ the vector of all variables with indices in
$\bm{\gamma}$. 
Note that while $\text{TC}(\bm X^2) = \text{DTC}(\bm X^2) = \text{II}(\bm X^2) = I(X_1;X_2)$, these quantities are in general different for $\bm X^n$ with $n\geq 3$. 
In effect, TC and DTC are non-negative metrics for all $n\geq 2$ that assess the overall strength of the correlations in different ways, as can be seen from the following identities~\cite[Cor.1]{rosas2019quantifying}:
\begin{align}
    \text{TC}(\bm X) =& \sum_{j=2}^n I(\bm X^{j-1};X_j),\label{eq:TC_exp}\\
    \text{DTC}(\bm X) =& I(\bm X^{n-1};X_n) + \sum_{j=2}^{n-1} I(\bm X^{j-1};X_j|\bm X_{j+1}^n).\label{eq:DTC_exp}
\end{align}
The interaction information is a signed metric which, in contrast with the O-information and the RSI, doesn't have a simple interpretation in terms of redundancies and synergies~\cite{williams2010nonnegative} (while having 
interpretations in terms of algebraic topology~\cite{baudot2015homological}). In the sequel we only use the interaction information for $n=3$, which for simplicity we denote as 
\begin{align}
    I(A;B;C) :=& \;\text{II}\big( (A,B,C) \big) 
    = I(A;B) - I(A;B|C).\label{eq:interaction_info}
\end{align}

Finally, note that conditional versions of all these quantities can be introduced by simply taking the average of the metric in question evaluated on the conditional quantity, akin to the definition of the conditional entropy (see e.g. Ref.~\cite[p.17]{CoverThomas06}).

\section{Basic properties}

This section presents a characterisation of the O-information and the RSI in terms of multivariate extensions of Shannon's mutual information. These characterisations, grouped into two lemmas, provide fundamental identities over which our main results build upon. Our first lemma expresses the O-information and RSI in terms of the TC and DTC.

\begin{lemma}\label{lemma:TCandDCT}
The following equalities hold:
\begin{align}
\Omega(\bm X) &= \normalfont{\text{TC}}(\bm X) - \normalfont{\text{DTC}}(\bm X),\label{eq:oinfo_TCDTC}\\
\normalfont{\text{RSI}}(\bm X;Y) &= \normalfont{\text{TC}}(\bm X) - \normalfont{\text{TC}}(\bm X|Y).\label{eq:rsi_TCDTC}
\end{align}
\end{lemma}
\begin{proof}
Direct from comparing \eqref{eq:def_Oinfo}, \eqref{eq:def_RSI}, \eqref{eq:def_TC}, and \eqref{eq:def_DTC}.
\end{proof}

This lemma implies that both the O-information and the RSI capture only high-order effects, as $\Omega(\bm X^2) = \text{RSI}(X_1;X_2) = 0$ and hence they give non-trivial results only when $n\geq 3$. The lemma also provides some first intuitions of what these quantities are measuring. The RSI captures to what extent correlations within $\bm X$ are created and/or destroyed by conditioning on $Y$. On the other hand, the O-information captures the differences between the correlations in $\bm X$ as measured by the TC and DTC, which --- thanks to Eqs.~\eqref{eq:TC_exp} and \eqref{eq:DTC_exp} --- can also be seen to be related to conditioning, but in a more complicated way. Our second lemma exploits this angle, expressing the O-information and RSI in terms of the interaction information for $n=3$ variables.

\begin{lemma}\label{lemma:decompositions}
    The following identities hold:
    \begin{align}
        \Omega(\bm X) &= \sum_{j=2}^{n-1} I(\bm X^{j-1}; X_j;\bm X_{j+1}^n), \label{eq:oinfo_dec}\\
        \normalfont{\text{RSI}}(\bm X;Y) &= \sum_{j=1}^{n-1} I(\bm X^{j-1}; X_j; Y), \label{eq:rsi_dec}
    \end{align}
    for $n\geq 3$ in Eq.~\eqref{eq:oinfo_dec} and $n\geq 2$ in Eq.~\eqref{eq:rsi_dec}.
\end{lemma}
\begin{proof}
Follows directly from Eqs.~\eqref{eq:TC_exp}, \eqref{eq:DTC_exp}, and~\eqref{eq:interaction_info} in combination with the results of Lemma~\ref{lemma:TCandDCT}.
\end{proof}

\begin{corollary}\label{cor:equal_n3}
    RSI and O-information coincide for the case of three variables:
    \begin{equation}
        \Omega(\bm X^3) = \normalfont{\text{RSI}}(\bm X^2;X_3) = I(X_1;X_2;X_3),
    \end{equation}
    but are in general different for $n\geq 4$.
\end{corollary}

While Corollary~\ref{cor:equal_n3} proves that the O-information and RSI are always equal for three variables and generally different for more, it is still not clear why is this the case. 
Section~\ref{sec:equalities} presents results that illuminate why and how these quantities differ for $n\geq 4$.

Although not strictly necessary for our analysis, here we present tight upper and lower bounds for the O-information and RSI. For this, we define $|\mathcal{X}^\text{max}| \coloneqq \max_{j=1}^n |\mathcal{X}_j|$ as the maximum of the cardinalities of the alphabets of all $X_i$'s, as well as $|\mathcal{T}| \coloneqq \min \big\{ |\mathcal{Y}|, |\mathcal{X}^\text{max}| \big\}$, and $|\mathcal{T}'| \coloneqq \min \big\{ |\mathcal{Y}|, \prod_{i=1}^n|\mathcal{X}_i|\big\}$.
 
\begin{lemma}\label{cor:bounds}
    The following bounds hold:
    \begin{align}
        (n-2) \log | \mathcal{X}^{\normalfont\text{max}} |  \geq& ~\Omega(\bm X) \geq -(n-2) \log | \mathcal{X}^{\normalfont\text{max}} | \label{eq:bounds_Oinfo},\\
        (n-1) \log | \mathcal{T} | \geq& ~\normalfont{\text{RSI}}(\bm X;Y) \geq -\log | \mathcal{T}' |. \label{eq:bounds_RSI}
    \end{align}
\end{lemma}
\begin{proof}
   Eq.~\eqref{eq:bounds_Oinfo} follow from \eqref{eq:oinfo_dec} after noticing that
   \begin{equation}
       |\mathcal{J}| \geq I(A;B;C) \geq -|\mathcal{J}|,\nonumber
   \end{equation}
   where $|\mathcal{J}| = \min\{|\mathcal{A}|,|\mathcal{B}|,|\mathcal{C}|\}$ is the size of the smallest alphabet among $A$, $B$, and $C$, which follows from Eq.~\eqref{eq:interaction_info} and the fact that $I(A;B;C)$ is symmetric in all its three arguments. The bounds in Eq.~\eqref{eq:bounds_RSI} follow directly from Eq.~\eqref{eq:def_RSI}.   
\end{proof}

It is easy to check these bounds are tight using generalisations of the copy and \texttt{xor} systems from Section~\ref{sec:preliminaries}. Let us first consider the `$n$-bit copy': a random binary vector $\bm X$ and random variable $Y$ such that $X_1$ is a fair coin flip and $X_1=X_2=\dots=X_n=Y$. It is direct to check that this yields $\Omega(\bm X) = n-2$ and $\text{RSI}(\bm X;Y) = n-1$, attaining the upper bounds. 
Now consider the `$n$-bit \texttt{xor}': a random binary vector $\bm X$ such that $X_1,\dots,X_{n-1}$ are i.i.d. fair coins and $X_n = \sum^{n-1}_{j=1} X_j \text{(mod 2)}$, which gives $\Omega(\bm X) = -(n-2)$. Another variation of this idea, where $X_1,\dots,X_n$ are all fair coins and $Y=\sum^{n}_{j=1} X_j \text{(mod 2)}$ gives $\text{RSI}(\bm X;Y) = -1$. These two types of systems --- copy and \texttt{xor} --- are extreme examples belonging to more general families of redundant or synergistic distributions explored further in Section~\ref{sec:log-likelihood}.

\section{Relating directed and undirected metrics}
\label{sec:equalities}

We now present our first main result: expressions for the O-information and the RSI in terms of each other which highlight similarities and differences between them. 
We first present an expression of the RSI in terms of the O-information.

\begin{proposition}\label{prop1a}
The $\text{RSI}$ can be expressed in terms of the O-information as follows:
\begin{align}
    \normalfont{\text{RSI}}(\bm X;Y) &= \Omega(\bm X,Y) - \Omega(\bm X|Y). \label{eq:prop1a}
\end{align}
\end{proposition}

\begin{proof}
We start by noting that, using Eq.~\eqref{eq:oinfo_dec}, one has that
    \begin{equation}
        \Omega(\bm X,Y) = I(\bm X^{n-1};X_n;Y) + \sum_{j=2}^{n-1} I(\bm X^{j-1};X_j;\bm X_{j+1}^n,Y).\label{eq:deriv1}
    \end{equation}
To proceed, we note that the interaction information possess the following chain rule:
\begin{align}
\label{eq:chain_rule_interaction}
    I(A;B;CD) &= I(A;B;C) + I(A;B;D|C),\nonumber
\end{align}
which can be verified directly from Eq.~\eqref{eq:interaction_info}. Using this chain rule on Eq.~\eqref{eq:deriv1}, one obtains that
    \begin{align}
        \Omega(\bm X,Y) \!
        =& \sum_{j=2}^n I(\bm X^{j-1};X_j;Y) \!+\! \sum_{j=2}^{n-1} I(\bm X^{j-1};X_j;\bm X_{j+1}^n|Y) \nonumber\\
        =& \:\text{RSI}(\bm X;Y) + \Omega(\bm X|Y),\nonumber
    \end{align}
    which is equivalent to Eq.~\eqref{eq:prop1a}. Above, the last equality follows from Eqs.~\eqref{eq:oinfo_dec} and \eqref{eq:rsi_dec}.
\end{proof}

This result illuminates what the RSI is measuring: the high-order effects in the whole system (given by the O-information) minus the high-order effects in the predictors induced by conditioning on the target variable $Y$. 
This explains, for example, why RSI and O-info are equal for $n=3$: if there are only two predictors, there cannot be high-order effects between them.

The next Corollary allows us to express the O-information in terms of the $\text{RSI}$. This result shows that the O-information is equivalent to seeing the system described by $\bm X$ via successive cuts, with each cut considering a different target variable and predictors. 

\begin{corollary}\label{cor:o_by_rsi}
The O-information can be expressed in terms of the RSI as follows:
\begin{align}
    \Omega(\bm X) 
    = \normalfont{\text{RSI}}(\bm X^{n-1};X_n) + \sum_{j=2}^{n-2} \normalfont{\text{RSI}}(\bm X^j;X_{j+1} | \bm X_{j+2}^n).\nonumber
\end{align}
\end{corollary}

\begin{proof} 
Let us first use Proposition~\ref{eq:prop1a} to show that
\begin{equation}
    \Omega(\bm X) 
    = \; \Omega(\bm X^{n-1}|X_n) + \text{RSI}(\bm X^{n-1};X_n).\nonumber
\end{equation}
Then, one can proceed by applying again Proposition~\ref{eq:prop1a} to the resulting conditional O-information, 
as follows:
\begin{align}
    \Omega(\bm X) 
    =& \; \Omega(\bm X^{n-1}|X_n) + \text{RSI}(\bm X^{n-1};X_n) \nonumber \\
    =& \; \Omega(\bm X^{n-2}|
    \bm X_{n-1}^n) + \text{RSI}(\bm X^{n-2};X_{n-1}|X_n) \nonumber\\
    &+  \text{RSI}(\bm X^{n-1};X_n) \nonumber \\
    =& \; \dots \nonumber \\
    =& \; \sum_{j=2}^{n-1} \text{RSI}(\bm X^{j};X_{j+1} | \bm X_{j+2}^n).\nonumber
\end{align}
\end{proof}

\section{High-order metrics as evidence towards synergistic or redundant models}
\label{sec:log-likelihood}

In this section we derive interpretations of RSI and O-information in terms of log-likelihood ratios, inspired by the Neyman-Pearson theory of hypothesis testing. To this end, the first step is to consider different classes of distributions that satisfy specific structures of conditional independence as determined by specific graphical models. For this, let us define the \textit{tail-to-tail} class given by
\begin{equation}
    \mathcal{C}_\text{T} \coloneqq \left\{q(\bm X,Y) : q(\bm X,Y) = q(Y) \prod_{j=1}^{n} q(X_j|Y) \right\},
\end{equation}
which corresponds to the class of all distributions where $X_1,\dots,X_{n}$ are conditionally independent given $Y$ (i.e. arrows go from $Y$ to each $X_i$). 
Similarly, let's also define the \textit{head-to-head} class given by
\begin{equation}
    \mathcal{C}_\text{H} \coloneqq \left\{q(\bm X,Y) : q(\bm X,Y) = q(Y|\bm X) \prod_{j=1}^{n} q(X_j) \right\},
\end{equation}
which corresponds to the class of all distributions where each $X_1,\dots,X_n$ are independent but $Y$ depends (potentially) on all of them --- i.e., arrows are coming from each variable $X_1,\dots,X_{n}$ towards $Y$. The copy and \texttt{xor} systems belong to 
$\mathcal{C}_\text{T}$ and $\mathcal{C}_\text{H}$, respectively. Probabilistic graphical models corresponding for both classes are shown in Figure~\ref{fig:pgm}. 

Our next result shows that the high-order interdependencies of distributions in $\mathcal{C}_\text{H}$ and $\mathcal{C}_\text{T}$ have opposite properties.

\begin{lemma}\label{lemma:cool}
    Let $p$ be the joint distribution of $(\bm X,Y)$. Then, 
    \begin{align}
        \normalfont{\text{RSI}}(\bm X;Y) = \normalfont{\text{TC}}(\bm X) &\geq 0 \qquad \text{if }p\in\mathcal{C}_{\normalfont{\text{T}}},\\
        \normalfont{\text{RSI}}(\bm X;Y) = - \normalfont{\text{TC}}(\bm X|Y) &\leq 0 \qquad \text{if }p\in\mathcal{C}_{\normalfont{\text{H}}}.
    \end{align}
\end{lemma}
\begin{proof}
    The result follows from Eq.~\eqref{eq:rsi_TCDTC} using the fact that $\text{TC}(\bm X|Y)= 0$ for $p\in\mathcal{C}_\text{T}$ and 
    $\text{TC}(\bm X)= 0$ for $p\in\mathcal{C}_\text{H}$.
\end{proof}

Intuitively, in tail-to-tail distributions any correlations within $\bm X$ are mediated by $Y$, and hence this information is necessarily redundant. 
On the contrary, in head-to-head distributions the $X_i$'s are (marginally) independent and may only become entangled when conditioned on $Y$, and hence these relationships are synergistic (as they involve at least three variables, $X_i$, $X_j$, and $Y$).

Let us now consider a dataset $\mathcal{S}_m$ of $m$ i.i.d. points of the form $\{(\bm x^{(k)},y^{(k)})\}_{k=1}^m$ sampled from a distribution $p(\bm X, Y)$. The likelihood of a distribution $q(\bm x,y)$ evaluated on the sample $\mathcal{S}_m$ is

\begin{equation}
    \Lambda(\mathcal{S}_m;q) \coloneqq \prod_{k=1}^m q(\bm x^{(k)},y^{(k)}).
\end{equation}

Using this definition, one can express the {\it generalized likelihood ratio test} (GLRT) as defined in Ref.~\cite[p.~86-96]{VanTrees} between two classes of models $\mathcal{C}$ and $\mathcal{C}'$, which is given by

\begin{equation}
    \lambda_{\mathcal{C}'}^{\mathcal{C}} (\mathcal{S}_m)
    \coloneqq \log \frac{ \sup_{r\in\mathcal{C}} \Lambda(\mathcal{S}_m;r)}{ \sup_{s\in\mathcal{C}'} \Lambda(\mathcal{S}_m;s)}.
\end{equation}

The GLRT is a well-known test for composite hypothesis testing.\footnote{When $\mathcal{C}$ and $\mathcal{C}'$ are point hypotheses, the Neyman-Pearson lemma shows that the GLRT is the optimal (i.e. most powerful) test. However, in general the GLRT cannot be shown to be an optimal test for composite hypotheses. See e.g. the discussion in Ref.~\cite{VanTrees}.} In fact, it can be interpreted as a comparison between the evidence towards accepting the hypothesis that $p\in\mathcal{C}$ versus $p\in\mathcal{C}'$ --- as $\lambda_{\mathcal{C}'}^{\mathcal{C}} (\mathcal{S}_m) = \lambda^{\mathcal{M}}_{\mathcal{C}'} (\mathcal{S}_m) - \lambda^{\mathcal{M}}_{\mathcal{C}} (\mathcal{S}_m)$ where $\mathcal{M}$ is the class of all distributions $p(\bm X, Y)$.\footnote{In the typical Neyman-Pearson setting, $\lambda^{\mathcal{M}}_{\mathcal{C}}$ quantifies the evidence in favour of the alternative hypothesis $p\in\mathcal{M}\setminus\mathcal{C}$ against the null hypothesis $p\in\mathcal{C}$.}
The following proposition establishes that the RSI corresponds to the GLRT between the model classes $\mathcal{C}_\text{T}$ and $\mathcal{C}_\text{H}.$

\begin{figure}
    \centering
    \begin{tikzpicture}

    \tikzstyle{mylatent} = [circle,fill=white,draw=black,inner sep=1pt,
minimum size=17pt, font=\fontsize{10}{10}\selectfont, node distance=1, anchor=center, align=center]

    \def\vsep{17mm}
    
\matrix[matrix of nodes,
  name=mat,
  nodes={anchor=center,
         align=center,
         font=\fontsize{10}{1}\selectfont},
  column sep=0.34cm,
  row sep=0.1cm]{

Tail-to-tail models $\mathcal{C}_\text{T}$

&

Head-to-head models $\mathcal{C}_\text{H}$

\\

\begin{tikzpicture}
\node[mylatent] (y1) {\footnotesize $X_1$};
\node[mylatent, right=0.3cm of y1] (y2) {\footnotesize $X_2$};
\node[mylatent, right=0.3cm of y2] (y3) {\footnotesize $\cdots$};
\node[mylatent, right=0.3cm of y3] (y4) {\footnotesize $X_n$};
\node[fit=(y1) (y2) (y3) (y4), inner sep=0] (naiveup) {};
\coordinate (p) at (0,+\vsep);
\node[mylatent] (y5) at (p -| naiveup) {\footnotesize $Y$};

\edge {y5} {y1.north};
\edge {y5} {y2.north};
\edge {y5} {y3.north};
\edge {y5} {y4.north};
\end{tikzpicture}

&

\begin{tikzpicture}
\node[mylatent] (u1) {\footnotesize $X_1$};
\node[mylatent, right=0.3cm of u1] (u2) {\footnotesize $X_2$};
\node[mylatent, right=0.3cm of u2] (u3) {\footnotesize $\cdots$};
\node[mylatent, right=0.3cm of u3] (u4) {\footnotesize $X_n$};
\node[fit=(u1) (u2) (u3) (u4), inner sep=0] (naiveup) {};
\coordinate (p) at (0,-\vsep);
\node[mylatent] (u5) at (p -| naiveup) {\footnotesize $Y$};

\edge {u1} {u5};
\edge {u2} {u5};
\edge {u3} {u5};
\edge {u4} {u5};
\end{tikzpicture}

\\

\small Predominantly redundant
&
\small Predominantly synergistic
\\
};

    \end{tikzpicture}
    \caption{Classes of probabilistic graphical models with predominantly redundant ($\mathcal{C}_\text{T}$) and synergistic ($\mathcal{C}_\text{H}$) interdependencies.}
    \label{fig:pgm}
\end{figure}
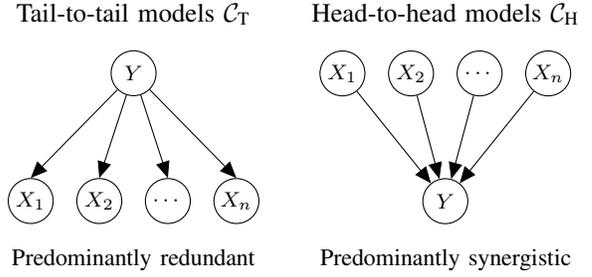

\begin{proposition}\label{prop_LLR}
Given a sequence of sets of i.i.d. samples $\mathcal{S}_m$, then the GLRT $\lambda_{\mathcal{C}_{\normalfont{\text{H}}}}^{\mathcal{C}_{\normalfont{\text{T}}}}/m$ converges in probability to the RSI:
\begin{equation}\label{eq:res_LRR}
\frac{1}{m}\lambda_{\mathcal{C}_{\normalfont{\text{H}}}}^{\mathcal{C}_{\normalfont{\text{T}}}}(\mathcal{S}_m) 
\xrightarrow[m \to \infty]{}
\normalfont{\text{RSI}}(\bm X;Y).
\end{equation}

\end{proposition}

\begin{proof}
Our strategy for proving this result will be to split $\lambda_{\mathcal{C}_\text{H}}^{\mathcal{C}_\text{T}}$ into two additive terms, show that each one converges, and then show that their addition gives the desired result. For this, let's use the shorthand notation
\begin{equation}
    \lambda_\text{T} \coloneqq \log \sup_{r\in\mathcal{C}_\text{T}} \Lambda(\mathcal{S}_m; r) \:\:\:\text{and}\:\:\:
    \lambda_\text{H} \coloneqq \log \sup_{s\in\mathcal{C}_\text{H}} \Lambda(\mathcal{S}_m; s),
\end{equation}
so that $\lambda_{\mathcal{C}_\text{H}}^{\mathcal{C}_\text{T}}(\mathcal{S}_m) = \lambda_\text{T} - \lambda_\text{H}$.

Let us first work out an expression for $\lambda_\text{T}$. For this, let us first note that

\begin{align}
    \lim_{m \to \infty} \frac{1}{m} \lambda_\text{T} 
    \stackrel{(a)}{=}& \sup_{q\in\mathcal{C}_\text{T}} \lim_{m \to \infty} \frac{1}{m} \sum_{k=1}^m \log q(\bm x^{(k)},y^{(k)}) \\\
    \stackrel{(b)}{=}& \sup_{q\in\mathcal{C}_\text{T}} \E[p]{\log q(\bm X, Y)} \\
    \stackrel{(c)}{=}& \sup_{q\in\mathcal{C}_\text{T}} \! \left\{\!\E[p]{\log q(Y)} + \!\sum_{j=1}^n \E[p]{\log q(X_j|Y)} \right\}. \nonumber 
\end{align}

where $(a)$ uses the fact that $\sup$ commutes with $\log$, $\lim$, and product; $(b)$ the weak law of large numbers; and $(c)$ the fact that
$q\in\mathcal{C}_T$.

From there, one can see that, by Gibbs' inequality, the suprema above are achieved when $q(Y) = p(Y)$ and $q(X_j|Y)=p(X_j|Y) \; \forall j$, yielding
\begin{align}
\lim_{m\to\infty}\frac{1}{m}\lambda_\text{T} 
&= -H(Y) - \sum_{j=1}^n H(X_j|Y).
\end{align}
By following an analogous derivation, one can show that
\begin{equation}
\lim_{m\to\infty}\frac{1}{m}\lambda_\text{H} = -H(Y|\bm X) - \sum_{j=1}^n H(X_j).
\end{equation}
The result then follows by taking the difference between these two limits.
\end{proof}

\begin{corollary}\label{cor:RSI_KL}
    \begin{align}
        \normalfont{\text{RSI}}(\bm X;Y) 
        &= \lim_{m\to\infty} \frac{1}{m} \lambda^{\mathcal{M}}_{\mathcal{C}_{\normalfont{\text{H}}}}(\mathcal{S}_m)
        - \lim_{m\to\infty} \frac{1}{m} \lambda^\mathcal{M}_{\mathcal{C}_{\normalfont{\text{T}}}}(\mathcal{S}_m)
        \label{eq:LRR_diff}\\[5pt]
        &= \;\inf_{q\in\mathcal{C}_{\normalfont{\text{H}}}}
        \KL{p}{q}
        - \inf_{q'\in\mathcal{C}_{\normalfont{\text{T}}}} 
        \KL{p}{q'},\label{eq:KL_diff}
    \end{align}
    where $\mathcal{M}$ denotes the collection of all distributions $p(\bm x,y)$.
\end{corollary}
\begin{proof}
    The first equality follows from Eq.~\eqref{eq:res_LRR} and the fact that $\lambda_{\mathcal{C}_\text{H}}^{\mathcal{C}_\text{T}} = \lambda^{\mathcal{M}}_{\mathcal{C}_\text{H}} - \lambda^{\mathcal{M}}_{\mathcal{C}_\text{T}}$. 
    To prove the second equality, note that
    \begin{align}
        \lim_{m\to\infty}\frac{1}{m} \lambda^{\mathcal{M}}_{\mathcal{C}_\text{H}}
        =& \sup_{s\in\mathcal{M}} 
        \E[p]{\log s(\bm X, Y)}
        - \!\sup_{r\in\mathcal{C}_\text{H}} 
        \E[p]{\log r(\bm X, Y)}\nonumber\\
        =& \inf_{r\in\mathcal{C}_\text{H}} \KL{p}{r}.\nonumber
    \end{align}
The desired result is obtained by following a similar derivation to show that $\frac{1}{m} \lambda^{\mathcal{M}}_{\mathcal{C}_\text{T}} \xrightarrow[m \to \infty]{} \inf_{s\in\mathcal{C}_\text{T}} \KL{p}{s}$.
\end{proof}

These results show that the $\text{RSI}$ can be understood as a comparison between the evidence towards accepting the hypothesis $p\in\mathcal{C}_\text{T}$ versus $p\in\mathcal{C}_\text{H}$. 
Additionally, Corollary~\ref{cor:RSI_KL} also brings a geometric interpretation: following the principles of information geometry~\cite{amari2016information}, $\inf_{q\in\mathcal{C}}\KL{p}{q}$ corresponds to the length of the M-projection of $p$ on the statistical manifold $\mathcal{C}$. Therefore, the $\text{RSI}$ compares the lengths of the projections of $p$ to either $\mathcal{C}_\text{T}$ and $\mathcal{C}_\text{H}$: if $p$ is closer to being tail-to-tail, then it is predominantly redundant and its RSI is positive; whereas if $p$ is closer to being head-to-head, then it is predominantly synergistic and its RSI is negative.

 Interestingly, note that if $p\in\mathcal{C}_\text{T}$ then combining Corollary~\ref{cor:RSI_KL} and Lemma~\ref{lemma:cool} one can find that
 \begin{equation}
 \text{RSI}(\bm X;Y) = \inf_{q\in\mathcal{C}_\text{H}} \KL{p}{q} =\text{TC}(\bm X),
 \end{equation}
 showing that in this case $\text{TC}(\bm X)$ is equal to the information-geometric projection of $p$ to $\mathcal{C}_\text{H}$. Conversely, if $p\in\mathcal{C}_\text{H}$ then 
\begin{equation}
 -\text{RSI}(\bm X;Y) = \inf_{q'\in\mathcal{C}_\text{T}} \KL{p}{q'} =\text{TC}(\bm X|Y),
 \end{equation}
so in that case $\text{TC}(\bm X|Y)$ measures the reciprocal projection.

Proposition~\ref{prop_LLR} allows us to develop an interpretation of the O-information as the sum of log-likelihood ratios. For this, we define the following two classes of models:
\begin{align}
    \mathcal{K}_\text{T}^j &= \Big\{q(\bm X) : q(\bm X) = q(X_j) q(\bm X^{j-1}|X_j) q(\bm X_{j+1}^n|X_j) \Big\},\nonumber\\
    \mathcal{K}_\text{H}^j &= \Big\{q(\bm X) : q(\bm X) = q(X_j|\bm X_{-j}) q(\bm X^{j-1}) q(\bm X_{j+1}^n) \Big\}.\nonumber
\end{align}
Each class corresponds to models that split the system into three parts ($\bm X^{j-1}$, $\bm X_{j+1}^n$, and $X_j$) and consider only redundant or synergistic relationships between them.

\begin{proposition}\label{cor:oinfo_LLR}
Given a sequence of sets of i.i.d. samples $\mathcal{S}_m$, then 
\begin{equation}
   \frac{1}{m}\sum_{j=2}^{n-1}\lambda_{\mathcal{K}_{\normalfont{\text{H}}}^j}^{\mathcal{K}_{\normalfont{\text{T}}}^j}(\mathcal{S}_m)
    \xrightarrow[m \to \infty]{}
    \Omega(\bm X).
\end{equation}
\end{proposition}
\begin{corollary}
\begin{align}
    \Omega(\bm X) 
    &= \sum_{j=2}^{n-1} \Bigg( \lim_{m\to\infty} \frac{1}{m} \lambda^{\mathcal{M}}_{\mathcal{K}^j_{\normalfont{\text{H}}}}(\mathcal{S}_m)
    - \lim_{m\to\infty} \frac{1}{m} \lambda^\mathcal{M}_{\mathcal{K}^j_{\normalfont{\text{T}}}} (\mathcal{S}_m)\Bigg)\nonumber\\
    &=\;\sum_{j=2}^{n-2} \Bigg(\inf_{q\in\mathcal{K}^j_{\normalfont{\text{H}}}}
        \KL{p}{q} 
    - \inf_{q'\in\mathcal{K}^j_{\normalfont{\text{T}}}} 
        \KL{p}{q'} \Bigg).\nonumber
\end{align}
\end{corollary}
\begin{proof}
Direct from applying the previous results to Eq.~\eqref{eq:oinfo_dec}, and using the fact that the RSI and the interaction information are equal when considering three variables.
\end{proof}

Therefore, like the $\text{RSI}$, the O-information corresponds to a balance between the evidence supporting the hypothesis that the distribution is dominated by redundant interdependencies (as the ones in $\mathcal{C}_\text{T}$) versus synergistic ones (as the distributions in $\mathcal{C}_\text{H}$). However, here the evidence is accumulated over a sequence of cuts, each considering a different variable $X_j$ as target, and splitting the rest of the system into two parts that are treated as predictors. That said, note that the expression on the right-hand-side of the identities in Corollary~\ref{cor:oinfo_LLR} does not depend on the ordering of the X's.

\section{Conclusion}

The results presented here provide a first view on the relationship between the RSI and the O-information, which are examples of directed and undirected metrics of high-order interdependencies. In particular, Proposition~\ref{prop1a} shows the the RSI is capturing the high-order effects on the predictors and targets that are not found in the predictors, and Corollary~\ref{cor:o_by_rsi} shows that the O-information corresponds to the sum of the RSI applied to different partitions of the system. 
Furthermore, Proposition~\ref{prop_LLR} and its corollaries illuminate the RSI and O-information in terms of generalised log-likelihood ratios between synergistic and redundant models. 
Overall, these results deepen our understanding of these metrics and provide new ways to interpret their values.

Although these are general-purpose metrics that can be applied to any system, we believe these are particularly promising for exploring intrinsic and extrinsic high-order effects in neural representations. For example, the results in this paper could be used to investigate the relationship between high-order relationships between the activations within a population of neurons (in the undirected sense) and in the ways such a population encodes an external stimulus (in the directed sense). We hope future experimental and modelling work can use these results to further elucidate the bases of neural computation and information processing.

\bibliographystyle{IEEEtran}
\bibliography{bib}

\end{document}